\documentclass{article}

\usepackage[utf8x]{inputenc}
\usepackage{ucs}
\usepackage{natbib}

\usepackage{fullpage}
\usepackage{amsmath}
\usepackage{amsfonts}
\usepackage{amssymb}
\usepackage{amsthm}
\usepackage{graphicx}
\usepackage{enumerate}
\usepackage{geometry}
\usepackage{subfigure}
\usepackage{hyperref}
\usepackage{listings,relsize}
\newcommand{\indexfonction}[1]{\index{#1@\texttt{#1}}}

\usepackage{tikz}
\usetikzlibrary{decorations.pathmorphing}
\usetikzlibrary{fit,shapes}
\tikzstyle{var}=[ellipse,thick,draw=black,minimum size=1.2cm]
\theoremstyle{definition}
\newtheorem{proposition}{Proposition}

\newtheorem{corollary}[proposition]{Corollary}

\newcommand{\greg}[1]{#1}
\newcommand{\vitto}[1]{#1}

\newcommand{\PP}{\mathbb{P}}

\def\independenT#1#2{\mathrel{\setbox0\hbox{$#1#2$}%
\copy0\kern-\wd0\mkern4mu\box0}}

\title{Hidden Markov Model Applications in Change-Point Analysis} 

\author{T. M. Luong, V. Perduca, and G. Nuel\\
MAP5 Laboratory, Paris Descartes University, Paris, France}

\date{}
\begin{document}

\maketitle

\section{Introduction}\label{intro}

The detection of change-points in heterogeneous sequences is a statistical challenge with many applications in fields such as finance, reliability, signal analysis, neurosciences and biology. A wide variety of literature exists for finding an ideal set of change-points for characterizing the data, in terms of both the number of change-points in a given sequence and their corresponding locations.

A conventional expression of the change-point problem is as follows: given a dataset $X_{1:n}=(X_1,X_2,\ldots,X_n)$ of real-valued observations, to find an ideal set of $K$ non-overlapping intervals where the observations are homogeneous within each. For $K$ segments, the change-point model expresses the distribution of $X$ given a segmentation $S_{1:n}=(S_1,\ldots,S_n) \in \mathcal{M}_K$ as:

\begin{equation}\label{eq:themodel}
\PP_{\theta}(X_{1:n} | S_{1:n})=\prod_{i=1}^{n} \beta_{S_i}\left(X_i\right) = \prod _{s=1}^K \prod_{i,S_i=s} \beta_s\left(X_i\right)
\end{equation}
where $\beta_s$ is the emission distribution of the observed data in segment $s$, $\theta=(\theta_1,\ldots,\theta_K)$ is the set of model parameters, $S_i$ is the segment index at position $i$, and $\mathcal{M}_K$ is the set of all possible combinations of $S$ for fixed $K\geqslant 2$ number of segments \greg{(for example, $S_{1:7}=1122233$ corresponds to $n=7$ observations divided into $K=3$ segments with two change-points: the first one between positions $2$ and $3$, and the second one between positions $5$ and $6$)}. We refer to this approach to the change-point model as \textit{segment-based}, later we discuss another common approach to change-point detection that we refer to as \textit{level-based}. For simplicity, denote $\PP_{\theta}(-)$ as $\PP(-)$ when the set of parameters is clear according to the context. 

\subsection{Current HMM algorithms in change-point analysis}

The hidden Markov model \citep[HMM, see][]{rabiner89} is a commonly used tool for inference in change-point analysis. While HMMs are a conventional feature in change-point methods in bioinformatics \citep{fridlyand04}, it is also widely used in diverse research areas such as speech recognition \citep{rabiner89}, facial recognition \citep{nefian98}, financial time series \citep{ge00}, inflation models \citep{chopin04}, music classification \citep{kimber97}, brain imaging \citep{zhang01}, climate research \citep{hughes99}, and network security \citep{cho03}.

\greg{Change-point analysis can be seen as a HMM where the data are the observations and the unknown segmentation the hidden states.} HMM adaptations can therefore identify change-points by observations where a switch in hidden states is most likely to occur. A convenient feature of the HMM approaches is in inferential procedures, such as estimating the posterior marginal state distribution $\PP(S|X)$. An efficient computation in linear time of this quantity uses classical forward-backward recursions \citep{durbin98}. The HMM estimation in mixture and change-point problems can be accomplished through the expectation-maximization (E-M) algorithm \citep{dempster77,bilmes98}, and MCMC methods, including reversible jump MCMC \citep{green95}, Gibbs sampling \citep{chib98}, and recursive algorithms \citep{scott02}. Other approaches which use sampling to characterize the uncertainty in a HMM given the data include particle filtering \citep{fearnhead03} and MCMC \citep{guha08}.  A summary of the inferential procedures involved in HMM estimation can be found in \citet{cappe05}. In addition to the algorithm of linear complexity presented in this chapter, two other exact algorithms for estimating posterior distributions include a frequentist \citep{guedon08} and Bayesian \citep{rigaill11} approach which use modified versions of the forward-backward algorithm.

Many schemes for estimating the number of hidden states in the HMM have also been investigated, which in general include several penalization criteria. These include a modified Bayes Information Criterion \citep{zhang07} to adjust for the number of states in previously fitted HMM as well as adaptive methods \citep{lavielle05,picard05} for estimating the location and number of change-points. 

Section \ref{cpinhmm} presents two different frameworks for applying HMM to change-point models, Section \ref{inference} provides a summary of two procedures for inference in change-point analysis, Section \ref{examples} provides two examples of the HMM methods on available data sets, Section \ref{genomics} provides a short summary of HMM and other change-point methods for current genomics studies, and Section \ref{conclusion} provides a short conclusion and discussion.

\section{The level- and segment- based change-point models as particular cases of the HMM formalisms}\label{cpinhmm}

We present a unifying HMM approach that can be adapted to many different approaches to the change-point problem. The known properties and current algorithms of the HMM allow for the efficient estimation of many quantities of interest. In turn, the results from the HMM may be used for model selection, or provide additional information about a given segmentation, such as the uncertainty of the change-points. The HMM framework also permits various extensions of the change-point model for further practical use. 

A typical HMM is defined by the joint probability distribution
\begin{equation}
\label{hmm}
\PP(X_{1:n},S_{1:n}) =\PP(S_1)\PP(X_1|S_1)\prod_{i=2}^n\PP(S_i|S_{i-1})\PP(X_i|S_i)
\end{equation}
where $n$ is the total number of observations, $X_{1:n}=(X_1,\ldots,X_n)$ is the vector of all the observation variables, $S_{1:n}=(S_1,\ldots,S_n)$ is the vector of all the hidden variables. Figure~\ref{fig:hmm} provides a simple example of the dependencies among variables. 

\begin{figure}
\centering{
\begin{tikzpicture}[>=latex,text height=1.5ex,text depth=0.25ex,scale=0.8, transform shape] 
   \draw (0,0) node (X1) [var] {$X_1$};
   \draw (2,0) node (X2) [var] {$X_2$};
   \draw (4,0) node (X3) [var] {$X_3$};
   \draw (6,0) node (X4) [var] {$X_4$};
   \draw (8,0) node (X5) [var] {$X_5$};
   \draw (0,2) node (S1) [var] {$S_1$};
   \draw (2,2) node (S2) [var] {$S_2$};
   \draw (4,2) node (S3) [var] {$S_3$};
   \draw (6,2) node (S4) [var] {$S_4$};
   \draw (8,2) node (S5) [var] {$S_5$};
   \path[->]
       (S1) edge[thick] (X1)
       (S2) edge[thick] (X2)
       (S3) edge[thick] (X3)
       (S4) edge[thick] (X4)
       (S5) edge[thick] (X5)
       (S1) edge[thick] (S2)
       (S2) edge[thick] (S3)
       (S3) edge[thick] (S4)
       (S4) edge[thick] (S5)
       ;
\end{tikzpicture}
\caption{HMM topology with $n=5$. For $i=1\ldots 5$, $S_i$ are the hidden states, and $X_i$ are the observed states.}
\label{fig:hmm}}
\end{figure}
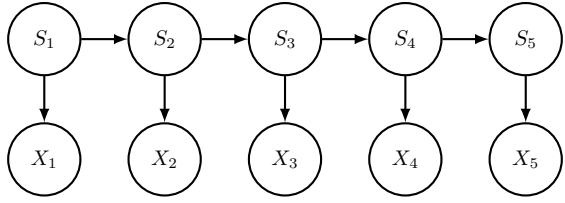
For \emph{homogeneous} HMMs define $\PP(S_1=s)=\mu(s)$, $\PP(S_i=s|S_{i-1}=r)=\alpha(r,s)$ for all $i=2,\ldots,n$ and $\PP(X_i=x|S_i=s;\theta)=\beta_s(x)$ for all $i=1,\ldots,n$. 

Given a \emph{evidence}, or some prior knowledge of the states of some variables, the computation of the posterior probabilities of the hidden variables is an important aspect of Hidden Markov modeling. For all observations $i=1,\ldots,n$, we introduce the formal notion of \emph{evidence} by considering $\mathcal{X}_i$ and $\mathcal{S}_i$, which are subsets of the two sets of all possible outcomes of $X_i$ and $S_i$, respectively.  The evidence is
\begin{equation}
\label{eq:ev}
\mathcal{E} = \{X_i\in\mathcal{X}_i,S_i\in\mathcal{S}_i,\, \mbox{ for all } i=1,\ldots,n.\}
\end{equation}
This evidence, which can comprise both observed information $X_i$ and/or known prior information $S_i$, provides a constraint on the set of posterior distributions. The unconstrained case occurs when $\mathcal{X}_i$ and $\mathcal{S}_i$ both contain the set of all possible states for $X_i$ and $S_i$, respectively, for all $i=1,\ldots,n$. 

Depending on the definition of the hidden state variables $S$, their corresponding their probability distributions \vitto{and the available evidence}, the HMM framework allows the definition of level- and segment-based models as follows. 

\subsection{Level-based model\vitto{: the standard evidence}}

In the level-based model, the hidden state $S_i$ pertains to the level (or underlying distribution) of observation $X_i$. This is an appropriate model when underlying properties can be shared between observations in different, non-adjacent segments. This HMM can be defined similarly to the classical HMM segmentation models by choosing the finite set of $L\geqslant 1$ levels, with $S \in \{1,2,\ldots,L\}^n$. With this level-based approach $K\geqslant L$, and transitions are possible between any pair of states.

Because all observation variables $X_i$ are observed and there are no constraints on the hidden states, the evidence is 
\begin{equation}
\label{eq:ev_l}
\mathcal{E}^\mathcal{L}=\{X_1=x_1,\ldots,X_n=x_n\}=\{X_{1:n}=x_{1:n}\}.
\end{equation}
This is the standard evidence usually found in most applications.

The transition matrix between hidden states can be assumed to be homogeneous, and is often parsimoniously parametrized. For example, consider:
\begin{equation}
\alpha(r,s)=\left\{
\begin{array}{ll}
1-\eta_r & \text{if $s=r$}\\
\frac{\eta_r}{L-1} & \text{if $s\neq r$}
\end{array}
\right.
\label{eq:trans}
\end{equation}
which uses only $L$ free parameters.

\subsection{Segment-based model}

In the segment-based model, the objective is to find the best partitioning $S \in \mathcal{M}_K$ of the data into $K$ non-overlapping intervals, where the hidden state $S_i$ pertains to the segment index of observation $X_i$. 

One set of constraints \citep{luong12} that allows the HMM to correspond \textit{exactly} to the above segment-based model defined in Equation~(\ref{eq:themodel}) is 
\begin{equation}
\label{eq:ev_s}
\mathcal{E}^\mathcal{S}=\{S_1=1,S_n=K,X_{1:n}=x_{1:n}\}
\end{equation}
where the transition only permits increments of $0$ or $+1$ of the segment index $S_i$, where $K$ is the fixed total number of segments. 

In order to obtain a uniform prior distribution of $S_{1:n}$ over all possible segmentations with $K$ segments $\mathcal{M}_{K}$, define the transition matrix over the state space $\{1,\ldots,K,K+1\}$ (where $K+1$ is a ``junk'' state) as follows:
$$
\alpha(r,s)=\left\{
\begin{array}{ll}
1-\eta & \text{if $r \leqslant K$ and $s=r$}\\
\eta & \text{if $r \leqslant K$ and $s=r+1$}\\
1 & \text{if $s=r=K+1$}\\
0 & \text{otherwise}
\end{array}
\right.
$$
where $\eta \in ]0,1[$ is a fixed number. Note that the particular value of $\eta$ does affect $\mathbb{P}(S)$ but has no effect whatsoever on $\mathbb{P}(S|\mathcal{E}^\mathcal{S})$. An arbitrary choice of $\eta=0.5$ is sufficient for practical computations.

\section{Inference in level- and segment-based HMMs}\label{inference}

\subsection{Variable elimination}
The forward-backward algorithm \citep{rabiner89,durbin98} efficiently solves many inference problems in HMMs;  before going into its technical details we illustrate the motivation and the very simple ideas behind it. 

Let us consider a HMM of length $n$ in which each hidden variable has $K$ possible states. Suppose we observe the standard evidence $\mathcal{E}=\{X_{1:n}=x_{1:n}\}$. \vitto{This is the specific expression for evidence~(\ref{eq:ev_l}) found in the level-based model; similar expressions also apply to the evidence in the segment-based model~(\ref{eq:ev_s}) and for the most general form of evidence~(\ref{eq:ev}).} 

An important problem in inference is to estimate the probability of this evidence, which can be accomplished by summing the joint distribution $\PP(S_{1:n},X_{1:n}=x_{1:n})$ over all the hidden variables:
$$
\PP(X_{1:n}=x_{1:n}) =  \sum_{S_1,\ldots,S_n} \PP(S_1,\ldots,S_n,X_{1:n}=x_{1:n}).
$$
In a naive approach, the first step would be to evaluate the joint probability for each possible value of $(S_1,\ldots,S_n)$ and then perform the summations explicitly. However this is a highly inefficient method as its total cost is $O(K^n)$. 

The exponential blowup is addressed by observing that the factors in the joint distribution~(\ref{hmm}) depend \vitto{each} on a small number of variables. A much more efficient algorithm is to evaluate expressions depending on these factors once and then cache the results, which avoids generating them multiple times. The basic tools which accomplish this are the factorization given in Eq.~(\ref{hmm}) and the distributive property. 

For example, in the case $n=3$:
\begin{multline}
\label{back_elim}
\PP(X_{1:n}=x_{1:n})=\\
\sum_{S_1}\PP(S_1)\PP(X_1=x_1|S_1)\underbrace{\sum_{S_2}\PP(S_2|S_1)\PP(X_2=x_2|S_2)\underbrace{\sum_{S_3}\PP(S_3|S_2)\PP(X_3=x_3|S_3)}_{B_2(S_2)}}_{B_1(S_1)}.       
\end{multline}

In practice, we choose an ordering of the hidden variables (here the \emph{backward} order $S_3<S_2<S_1$) and then rearrange all factors in order, so that all the factors depending on $S_3$ are the furthest on the right. Then all remaining factors depending on $S_2$ are placed to the left, and similarly for $S_1$. 
This makes it possible to eliminate one variable after another, according to the initial order. For each eliminated variable, we obtain a quantity which we cache and use in order to eliminate the ensuing variable. These quantities are also called \emph{messages} and the outlined procedure provides a recursive method to compute them. In the case of the backward ordering, the messages are called \emph{backward quantities}.

The computation of $B_2(S_2)$ for all the values of $S_2$ requires the sum of $K$ terms (one for each value of $S_3$); thus resulting in $O(K^2)$ operations. Similarly, computing $B_1$ requires $O(K^2)$ operations. Because there are $n$ hidden variables the resulting total cost is $O(nK^2)$ \greg{which is a dramatic improvement over the naive $O(K^n)$ complexity}. 

A central aspect of variable elimination is the initial ordering of these hidden variables. For instance, consider the elimination order $S_2<S_3<S_1$, then: 
\begin{multline*}
\PP(X_{1:n}=x_{1:n})=\\
\sum_{S_1}\PP(S_1)\PP(X_1=x_1|S_1)\underbrace{\sum_{S_3}\PP(X_3=x_3|S_3)\underbrace{\sum_{S_2}\PP(S_3|S_2)\PP(S_2|S_1)\PP(X_2=x_2|S_2)}_{C(S_3,S_1)}}_{D(S_1)}.       
\end{multline*}
The cost of eliminating $S_2$ is $O(K^3)$, with $K$ terms summed for each value of the pair $(S_3,S_1)$. As a result, the resulting complexity is $O(nK^3)$. 

Another ordering which leads to a $O(nK^2)$ cost is the \emph{forward} ordering $S_1<S_2<S_3$: 
\begin{multline*}
\PP(X_{1:n}=x_{1:n})=\\
\sum_{S_3}\PP(X_3=x_3|S_3)\underbrace{ \sum_{S_2}\PP(S_3|S_2)\PP(X_2=x_2|S_2)\underbrace{\sum_{S_1}\PP(S_2|S_1)\PP(X_1=x_1|S_1)\PP(S_1)}_{E_2(S_2)}}_{E_3(S_3)}. 
\end{multline*}

$E_2(S_2)$ and $E_3(S_3)$ are closely related to the \emph{forward} quantities defined in Section~\ref{subsec:F/B}: 
having defined $F_1(S_1):= \PP(X_1=x_1|S_1)\PP(S_1)$:
$$
E_2(S_2)=\sum_{S_1}\PP(S_2|S_1)F_1(S_1),
$$ 
and 
$$
F_2(S_2) :=\PP(X_2=x_2|S_2)E_2(S_2)=\PP(X_2=x_2|S_2)\sum_{S_1}\PP(S_2|S_1)F_1(S_1)
$$
and 
$$
F_3(S_3) :=\PP(X_3=x_3|S_3)E_3(S_3)=\PP(X_3=x_3|S_3)\sum_{S_2}\PP(S_3|S_2)F_2(S_2).
$$ 

At this point there are two alternative formulae for solving our initial inference problem based on the recursive forward and backward quantities:
$$
\PP(X_{1:n}=x_{1:n}) = \sum_{S_3}F_3(S_3),  
$$
and also, from Eq.~(\ref{back_elim}):
$$
\PP(X_{1:n}=x_{1:n}) = \sum_{S_1}F_1(S_1)B_1(S_1).
$$

The next section explains the use of forward and backward quantities to compute the posterior probabilities $\PP(S_i|X_{1:n}=x_{1:n})$ and other useful distributions. For simplicity we will consider homogeneous HMMs. 

\subsection{Forward-backward algorithm} 
\subsubsection{Level-based model\vitto{: the standard forward-backward algorithm}}\label{subsec:F/B}



In the level-based case, the evidence~(\ref{eq:ev_l}) is the standard evidence found in many HMMs applications. 
\vitto{The inference problem is then solved by the classical recursive formulae used for most applications. }

Define the forward and backward quantities:
$$
F_i^\mathcal{L}(s):=\PP(S_i=s,X_{1:i}=x_{1:i})
$$
and 
$$
B_i^\mathcal{L}(s):=\PP(X_{i+1:n}=x_{i+1:n}|S_i=s),
$$ 
for all $i \in \{1,\ldots,n\}$ under the convention that $B_n^\mathcal{L} \equiv 1$. These probability distributions are \vitto{exactly} the standard forward-backward quantities found in most applications.

The recursive computation of the forward and backward quantities requires the following results: 
\begin{proposition}
\label{prop:fund_results}
For all $i=1,\ldots,n$ and for each value $s$ of the hidden states:
\begin{equation}
\label{eq:separator}
\PP(S_i=s,\mathcal{E}^\mathcal{L})=F_i^\mathcal{L}(s)B_i^\mathcal{L}(s),
\end{equation}
where $\mathcal{E}^\mathcal{L}$ is defined in~(\ref{eq:ev_l}). Moreover for all $i \in \{2,\ldots,n \}$, for each pair $(r,s)$ and for each observation $x$:
\begin{equation}
\label{eq:clique}
\PP(S_{i-1}=r,S_{i}=s,\mathcal{E}^\mathcal{L})=F^\mathcal{L}_{i-1}(r)\alpha(r,s)\beta_{s}(x)B_i^\mathcal{L}(s)
\end{equation}
\end{proposition}

\begin{proof}
We start by proving Eq.~(\ref{eq:separator}), by applying the chain rule and the conditional independence of the events $\{X_{i+1:n}=x_{i+1:n}\}$ and $\{X_{1:i}=x_{1:i}\}$ given $\{S_i=s\}$:
\begin{multline*}
\PP(S_i=s,\mathcal{E}^\mathcal{L})=\PP(S_i=s,X_{1:i}=x_{1:i},X_{i+1:n}=x_{i+1:n})=\\
\PP(S_i=s,X_{1:i}=x_{1:i})\PP(X_{i+1:n}=x_{i+1:n}|S_i=s,X_{1:i}=x_{1:i})=\\
\PP(S_i=s,X_{1:i}=x_{1:i})\PP(X_{i+1:n}=x_{i+1:n}|S_i=s).
\end{multline*}

Similarly for Equation~(\ref{eq:clique}):
\begin{multline*}
\PP(S_{i-1}=r,S_i=s,\mathcal{E}^\mathcal{L})=\PP(S_{i-1}=r,X_{1:i-1}=x_{1:i-1},S_i=s,X_{i}=x_i,X_{i+1:n}=x_{i+1:n})=\\
\PP(S_{i-1}=r,X_{1:i-1}=x_{1:i-1})\PP(S_{i}=s|S_{i-1}=r)\times\\
\PP(X_i=x_i|S_i=s)\PP(X_{i+1:n}=x_{i+1:n}|S_i=s).
\end{multline*}
\end{proof}

This proposition establishes all the classical results for inference in HMMs.

\begin{corollary}[Forward and backward recursions]
The forward quantities can be computed iteratively starting from $F_1^\mathcal{L}(s)=\mu(s)\beta_s(x_1)$ for all $i=2,\ldots,n$ with 
\begin{equation}
\label{eq:forw}
F_{i}^\mathcal{L}(s) =\sum_{r}F_{i-1}^\mathcal{L}(r)\alpha(r,s)\beta_s(x_i).
\end{equation}
The backward quantities can be computed recursively from $B_n^\mathcal{L}\equiv1$ for all $i=n-1,\ldots,1$ with 
\begin{equation}
\label{eq:back}
B_{i-1}^\mathcal{L}(r)=\sum_{s}\alpha(r,s)\beta_s(x_i)B_i^\mathcal{L}(s).
\end{equation}
\end{corollary}

\begin{proof}
In order to obtain the forward Equation~(\ref{eq:forw}), apply Equations~(\ref{eq:separator}) and~(\ref{eq:clique}) respectively to the right and left hand side of
$$
\PP(S_i=s,\mathcal{E}^\mathcal{L})=\sum_{r}\PP(S_{i-1}=r,S_i=s,\mathcal{E}^\mathcal{L}).
$$
A similar argument holds for the backward Equation~(\ref{eq:back}).
\end{proof}

The following useful result is a straightforward consequence of Proposition~\ref{prop:fund_results}:
\begin{corollary}[Posterior probabilities]
For each $i=1,\ldots,n$, the posterior state probabilities are
$$
\PP(S_i=s|\mathcal{E}^\mathcal{L})=\frac{F_i^\mathcal{L}(s)B_i^\mathcal{L}(s)}{\PP(\mathcal{E}^\mathcal{L})}\label{state}
\greg{
\quad\text{and}\quad
}
\greg{
\PP(S_{i-1}=r,S_i=s|\mathcal{E}^\mathcal{L})=
\frac{F_{i-1}^\mathcal{L}(r)\alpha(r,s)\beta_s(x_i)B_i^\mathcal{L}(s)}{\PP(\mathcal{E}^\mathcal{L})}
}
$$ 
where the probability of the evidence is 
$$
\PP(\mathcal{E}^\mathcal{L})=\sum_s F_i^\mathcal{L}(s)B_i^\mathcal{L}(s),
$$
for an arbitrary fixed $i\in\{1,\ldots,n\}$; in particular $\PP(\mathcal{E}^\mathcal{L})=\sum_sF_n^\mathcal{L}(s)$. 
\end{corollary}

The following corollary is easily proven and makes it possible to sample the joint distribution of the hidden variables recursively:

\begin{corollary}[Forward and backward sampling]
The joint distribution of $S_{1:n}$ conditionally to $\mathcal{E}^\mathcal{L}$ is a Markov chain whose transitions are given by 
$$
\PP(S_i=s|S_{i-1}=r,\mathcal{E}^\mathcal{L})=\frac{\alpha(r,s)\beta_s(x_i)B_{i}^\mathcal{L}(s)}{B_{i-1}^\mathcal{L}(r)}\label{transition}
$$
in the forward direction, and by 
$$
\PP(S_{i-1}=r|S_i=s,\mathcal{E}^\mathcal{L})=\frac{F_{i-1}^\mathcal{L}(r)\alpha(r,s)\beta_s(x_i)}{F_i^\mathcal{L}(s)}
$$
in the backward direction.
\end{corollary}

\subsubsection{Forward-backward algorithm for segment-based model}
\label{sec:seg}

\vitto{For the sake of completeness}, we present the results of the inference problem for the HMMs constrained by the evidence~(\ref{eq:ev_s}). The formulae presented here are easily obtained by modifying the standard formulae proved in section~\ref{subsec:F/B} and are particular cases of the forward-backward algorithm for HMMs conditioned on the general evidence~(\ref{eq:ev}). In particular, the formulae from the previous section still hold for $i\notin\{1,n\}$. In these two special cases consider the constraints $S_1=1$ and $S_n=K$, which are imposed by adding the multiplicative constants $\mathbf{1}_{\{S_1=1\}}$ or $\mathbf{1}_{\{S_n=K\}}$ to the formulae\footnote{Given an event $\mathcal{A}$, $\mathbf{1}_{\mathcal{A}}=1$ iff $\mathcal{A}$ is true.} .  

Define the forward quantities as
$$
F_i^\mathcal{S}(s):=\PP(S_1=1,S_i=s,X_{1:i}=x_{1:i})
$$
for $i\leq n-1$ and 
$$
F_n^\mathcal{S}(s):=\PP(S_1=1,S_n=s,S_n=K,X_{1:n}=x_{1:n})=\mathbf{1}_{\{s=K\}}\PP(S_1=1,S_n=s,X_{1:n}=x_{1:n}).
$$
Define the backward quantities as
$$
B_i^\mathcal{S}(s):=\PP(X_{i+1:n}=x_{i+1:n},S_n=K|S_i=s),
$$ 
for all $i \in \{1,\ldots,n\}$, with the convention that $B_n^\mathcal{S} \equiv 1$. The corresponding equations for ~(\ref{eq:separator}) and~(\ref{eq:clique}) become $\PP(S_i=s,\mathcal{E}^\mathcal{S})=F_i^\mathcal{S}(s)B_i^\mathcal{S}(s)$ and 
$$
\PP(S_{i-1}=r,S_{i}=s,\mathcal{E}^\mathcal{S})=F_{i-1}^\mathcal{S}(r)\alpha(r,s)\beta_{s}(x)B_i^\mathcal{S}(s)
$$
for all $i=1,\ldots,n-1$ and 
$$
\PP(S_{n-1}=r,S_{n}=s,\mathcal{E}^\mathcal{S})=\mathbf{1}_{\{s=K\}}F_{n-1}^\mathcal{S}(r)\alpha(r,s)\beta_{s}(x)B_n^\mathcal{S}(s).
$$

The forward quantities can be computed iteratively starting from $F_1^\mathcal{S}(s)=\mathbf{1}_{\{s=1\}}\mu(s)\beta_s(x_1)$ for all $i=2,\ldots,n-1$ with 
\begin{equation*}
\label{eq:forw_seg}
F_{i}^\mathcal{S}(s) =\sum_{r}F_{i-1}^\mathcal{S}(r)\alpha(r,s)\beta_s(x_i),
\end{equation*}
and $F_n^\mathcal{S}(s)=\mathbf{1}_{\{s=K\}}\sum_{r}F_{n-1}^\mathcal{S}(r)\alpha(r,s)\beta_s(x_i)$. The recursions for computing the backward quantities are exactly the same as in the level-based setting:
\begin{equation*}
\label{eq:back_seg}
B_{i-1}^\mathcal{S}(r)=\sum_{s}\alpha(r,s)\beta_s(x_i)B_i^\mathcal{S}(s).
\end{equation*}

The preceding constraints result in less computations due to the sparse transition matrix between hidden states, as $\alpha(r,s)=0$ if $s-r \notin (0,1)$ \greg{leading to a $O(Kn)$ complexity}. In the constrained model we note that the probability of the evidence can be computed, for instance, with $\PP(\mathcal{E}^\mathcal{S})=F_1^\mathcal{S}(1)B_1^\mathcal{S}(1)$.

\subsection{Extensions}

Extending the previously described recursive formulae and results to more general evidence forms is a straightforward process. For instance, suppose that we observe $X_i=x_i$ for all $i$ except for $i=5$, which we will consider as missing. In this situation, the evidence is $\mathcal{E}=\{X_{1:4}=x_{1:4},X_{6:n}=x_{6:n}\}$. Define the forward and backward quantities exactly as in the level-based case with the only difference that for $i\geq 5$, $F_i(s)=\PP(S_i=s,X_{1:4}=x_{1:4},X_{6:i}=x_{6:i})$ and for $i\leq4$ $B_i(s)=\PP(X_{i+1:4}=x_{i+1:4},X_{6:n}=x_{6:n}|S_i=6)$. In practice, all the results from the level-based case still hold if we substitute $\beta_s(x_5)$ by $1$ in the formulae expressing $\PP(S_{4}=r,S_5=s,\mathcal{E})$, $F_5(s)$, $B_{4}(r)$, $\PP(S_5=s|S_{4}=r,\mathcal{E})$ and $\PP(S_4=r|S_5=s,\mathcal{E})$.

Another situation is all variables $(X_1,\ldots,X_n)$ being observed with additional partial prior knowledge about the hidden variables. Suppose for instance $X_{1:n}=x_{1:n}$ and, say, $S_7\neq2$. The evidence is $\mathcal{E}=\{X_{1:n}=x_{1:n},S_7\neq2\}$ and we define the forward and backward quantities as before with the following exceptions: $F_i(s)=\PP(X_{1:i}=x_{1:i},S_7\neq2,S_i=s)$ for $i\geq 7$ and $B_i(s)=\PP(X_{i+1}=x_{i+1},S_7\neq 2|S_i=s)$ for $i<7$. All the results continue to hold by substituting $\alpha(r,s)$ by $\mathbf{1}_{s\neq2}\cdot\alpha(r,s)$ in the formulae expressing $\PP(S_{6}=r,S_7=s,\mathcal{E})$, $F_7(s)$, $B_{6}(r)$, $\PP(S_7=s|S_{6}=r,\mathcal{E})$ and $\PP(S_6=r|S_7=s,\mathcal{E})$.


The forward-backward algorithm is also known as the \emph{product-sum algorithm} and is a particular case of the \emph{message propagation} algorithm for Bayesian networks \greg{\citep{koller2009probabilistic}}. As previously seen, its basic mechanism is the simple distributive property of multiplication over addition. 

Another common inferential problem in HMMs is in finding a set of variables with the largest joint state probability and their corresponding marginal probabilities. In general, the set which maximizes the joint state probabilities may not exactly match the set that maximizes each marginal probability separately. The \emph{max-sum} algorithm (also known as the \emph{Viterbi} algorithm) addresses this issue by replacing summations by maximizations in all the above formulae \citep{viterbi67, rabiner89}. The previous results continue to hold because the multiplication is distributive over the max operator: $\max(ab,ac)=a\max(b,c)$.

\subsection{Floating-point issues}


Underflow is a common issue when using forward and/or backward recursions in floating-point arithmetic. Indeed, the magnitude of \greg{the forward quantity} $F_i$ decreases geometrically with $i$ and can be smaller than the smallest machine float (ex: $2.23\times 10^{-308}$ for C++ double-precision on i686 architecture) which is an architecture-dependent threshold. As suggested in \citep[pages 272-273]{rabiner89}, one solution consists in keeping track of a rescaling parameter (typically stored in log-scale) for each $i$. To improve this inefficient approach, in terms of both time and memory, we suggest to use log-scale computation through $\log F_i$ and $\log B_i$ \greg{both for level- and segment-based models}.

For both expressions of $F_i$ and $B_i$ we recommend the initial precomputation of $\log \alpha(r,s)$ and $\log \beta_s(x_i)$ for all $r,s,i$. The forward and backward recursions become:
$$
\log F_i(s)= \mbox{logsum} [ \log F_{i-1}(\cdot)+\log\alpha(\cdot,s)+\log \beta_s(x_i) ]
$$
$$
\log B_{i-1}(r)= \mbox{logsum} [ \log\alpha(r,\cdot)+\log \beta_\cdot(x_i)+\log B_{i}(\cdot) ]
$$
where $\mbox{logsum}$ is a function of any real vector $z=(z_1,z_2,\ldots z_k)$. For example if $z_1\leqslant z_2 \leqslant \ldots z_k$ and $\mbox{logsum}(z_1,z_2,\ldots,z_k)=\log \sum_j \exp(z_i)$:
$$
\mbox{logsum}(z_1,z_2,\ldots,z_k)=z_1+\mbox{log1p}[\exp(z_2-z_1)+\ldots+\exp(z_k-z_1)]
$$
where $\mbox{log1p}(u)=\log(1+u)$, which is useful in floating-point arithmetic when $u$ is small.

\subsection{E-M algorithm}

Let $\mathcal{E}$ be either $\mathcal{E}^\mathcal{L}$ or $\mathcal{E}^\mathcal{S}$ or a more general evidence. The Expectation-Maximization (EM) algorithm consists of repeating iteratively the two following steps:
\begin{itemize}
\item {\emph Expectation (E)}: compute $Q(\theta'|\theta)=\sum_{S} \mathbb{P}_\theta(S|\mathcal{E}) \log \mathbb{P}_{\theta'}(S,\mathcal{E})$;
\item {\emph Maximization (M)}: update $\theta$ with $\tilde\theta=\arg \max_{\theta'} Q(\theta'|\theta)$.
\end{itemize}

The resulting update formulae depend on the model considered (level- or segment-based) and on the nature of the emission distribution. In this section, we consider only the two following emission models:
\begin{itemize}
\item {\emph {normal homoscedastic}}: for all $x \in \mathbb{R}$ and hidden state $s$,  $\beta_s(x)=\varphi((x-\mu_s)/\sigma)$ where $\mu_s \in \mathbb{R}$, $\sigma>0$, with $\varphi(z)=\exp(z^2/2)/\sqrt{2\pi}$ for all $z \in \mathbb{R}$;
\item {\emph {Poisson}}: for all $k \in \mathbb{N}$ and hidden state $s$, $\beta_s(k)=\exp(-\lambda_s)\lambda_s^ k/k!$ where $\lambda_s>0$. 
\end{itemize}

During the E-step, forward-backward recursions obtain $\mathbb{P}_\theta(S| \mathcal{E})$ as a heterogeneous Markov chain. In practice, it is only necessary to compute $\mathbb{P}_\theta(S_i|\mathcal{E})=F_i(S_i)B_i(S_i)/\mathbb{P}_\theta(\mathcal{E})$ and $\mathbb{P}_\theta(S_{i-1},S_{i}|\mathcal{E})=F_{i-1}(S_{i-1})\alpha(S_{i-1},S_i)\beta_{S_i}(x_i)B_i(X_i)/\mathbb{P}_\theta(\mathcal{E})$.

For both level- and segment- based model, the M-step is the same for the emission part of $Q(\theta'|\theta)$:
$$
\tilde{\mu}_s=\tilde{\lambda}_s=\frac{\sum_{i=1}^n X_i \mathbb{P}_\theta(S_i=s|\mathcal{E})}{\sum_{i=1}^n \mathbb{P}_\theta(S_i=s|\mathcal{E})}
\quad\text{and}\quad
\tilde{\sigma}^2=\frac{\sum_{i=1}^n \sum_s (X_i-\tilde{\mu}_s)^2 \mathbb{P}_\theta(S_i=s|\mathcal{E})}{n}.
$$

For the transition part, there is no parameter to update in the segment-based case, and for the level-based case (with the transition matrix defined \greg{in Eq.\ref{eq:trans}}):
$$
\tilde{\eta}_r=\frac{
\sum_{r\neq s} \sum_{i=2}^ n \mathbb{P}_\theta(S_{i-1}=r,S_{i}=s|\mathcal{E})
}
{
\sum_{i=2}^ n \mathbb{P}_\theta(S_{i-1}=r|\mathcal{E})
}.
$$


\subsection{Change-point estimation}

\subsubsection{Level-based model}

Change-points are identified based on observations where transitions are most likely to occur. In a HMM framework, we assess the probability that $i$ is any change-point between two different hidden states, regardless of the ordering within the set of change-points. 

The posterior probability of any change-point occurring at observation $i$, or such that $S_i \neq S_{i+1}$, is:
\begin{align*}
\nonumber \PP(CP=i|\mathcal{E}^\mathcal{L})&=\underset{r\neq s}{\sum}\PP(S_i=r,S_{i+1}=s|\mathcal{E}^\mathcal{L})\\
\nonumber&=\underset{r\neq s}{\sum}\PP(S_{i+1}=s|S_i=r,\mathcal{E}^\mathcal{L})\PP(S_i=r|\mathcal{E}^\mathcal{L})\\
&=\underset{r\neq s}{\sum}\frac{F_{i}^\mathcal{L}(r) \alpha(r,s) \beta_s(x_{i+1})B_{i+1}^\mathcal{L}(s)}{F_1^\mathcal{L}(1)B_1^\mathcal{L}(1)}
\end{align*}
for $i=1,\ldots,n-1$ and $r,s=1,\ldots,L$\vitto{, where the last equality follows from Section~\ref{subsec:F/B}}.


Unlike the segment-based approach, the level-based approach does not allow to obtain the the distribution of the specific $r^{th}$ change-point, but only the marginal probability of a change-point at a given position. It is however possible to derive the distribution of the  $r^{th}$ change-point in the level-based approach as the waiting time of a regular expression in a heterogeneous Markov chain \citep{aston07}.

\subsubsection{Segment-based model}
\label{sec:segment_based}

Define the location of the $r^{th}$ change-point as $CP_r$. Under this definition, the posterior probability of the $r^{th}$ change at observation $i$ is $\PP(CP_r=i)=\PP(S_i=r,S_{i+1}=s|\mathcal{E}^\mathcal{S})$, where $s=r+1$. In other words the $r^{th}$ change-point is the last observation of segment $r$ before segment $r+1$, each of whom consists of contiguous, homogeneous observations.

Using the formulae provided in Section~\ref{sec:seg}, the posterior probability of the $r^{th}$ change-point occurring after observation $i$, where $s=r+1$ is:
\begin{align*}
\nonumber \PP(CP_r=i|\mathcal{E}^\mathcal{S})&=\PP(S_i=r,S_{i+1}=s|\mathcal{E}^\mathcal{S})\\
\nonumber&=\PP(S_{i+1}=s|S_i=r,\mathcal{E}^\mathcal{S})\PP(S_i=r|\mathcal{E}^\mathcal{S})\\
&=\frac{F_{i}^\mathcal{S}(r) \alpha(r,r+1) \beta_{r+1}(x_{i+1})B_{i+1}^\mathcal{S}(r+1)}{F_1^\mathcal{S}(1)B_1^\mathcal{S}(1)}. 
\end{align*}
for $i=1,\ldots,n-1$ and $r=1,\ldots,K-1$.

\section{Examples}\label{examples}

\subsection{British coal mining disaster data set}

\begin{figure*}
\subfigure[Posterior state probability]{
\includegraphics[scale=0.35]{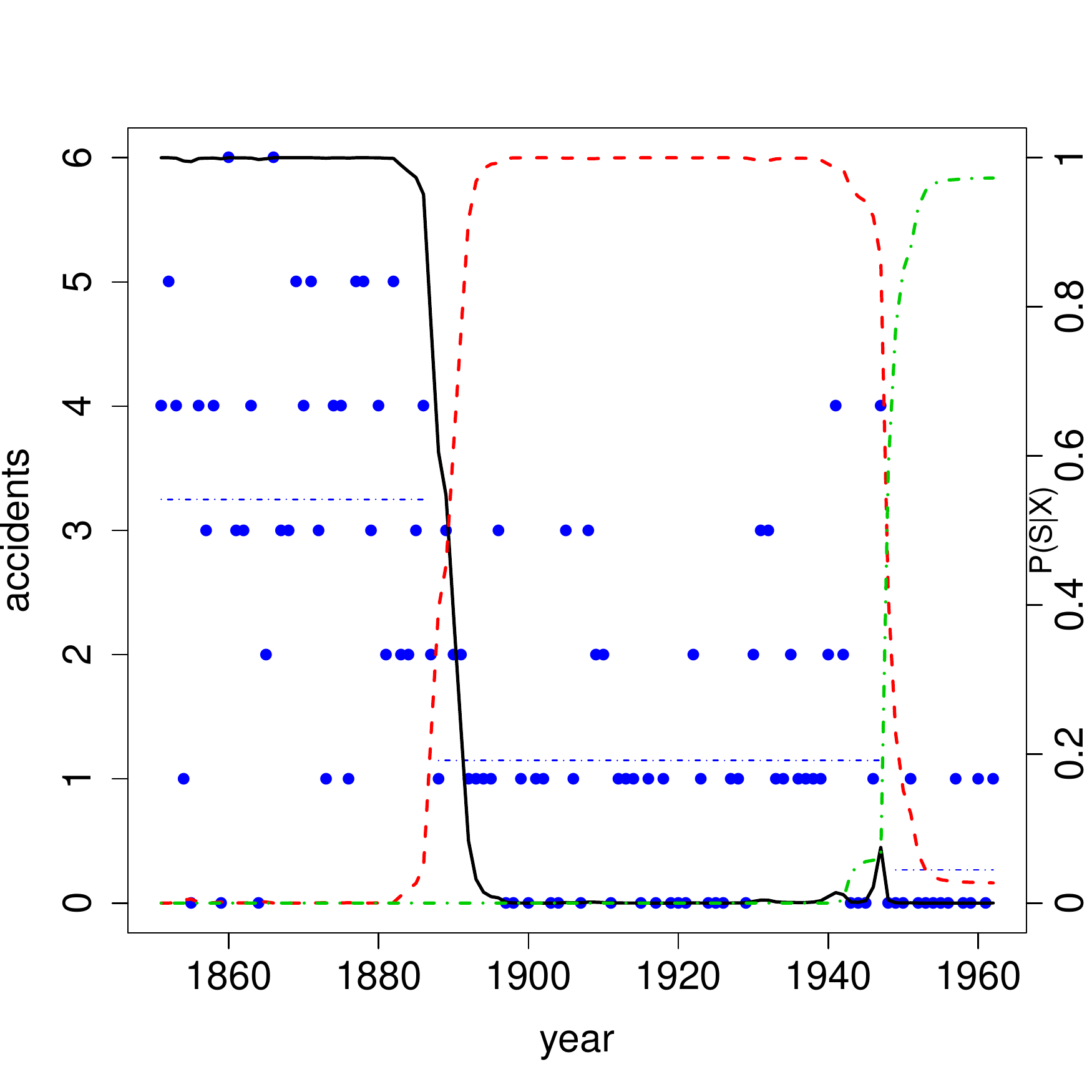}
\label{cls}}       
\subfigure[Posterior change-point probability]{
\includegraphics[scale=0.35]{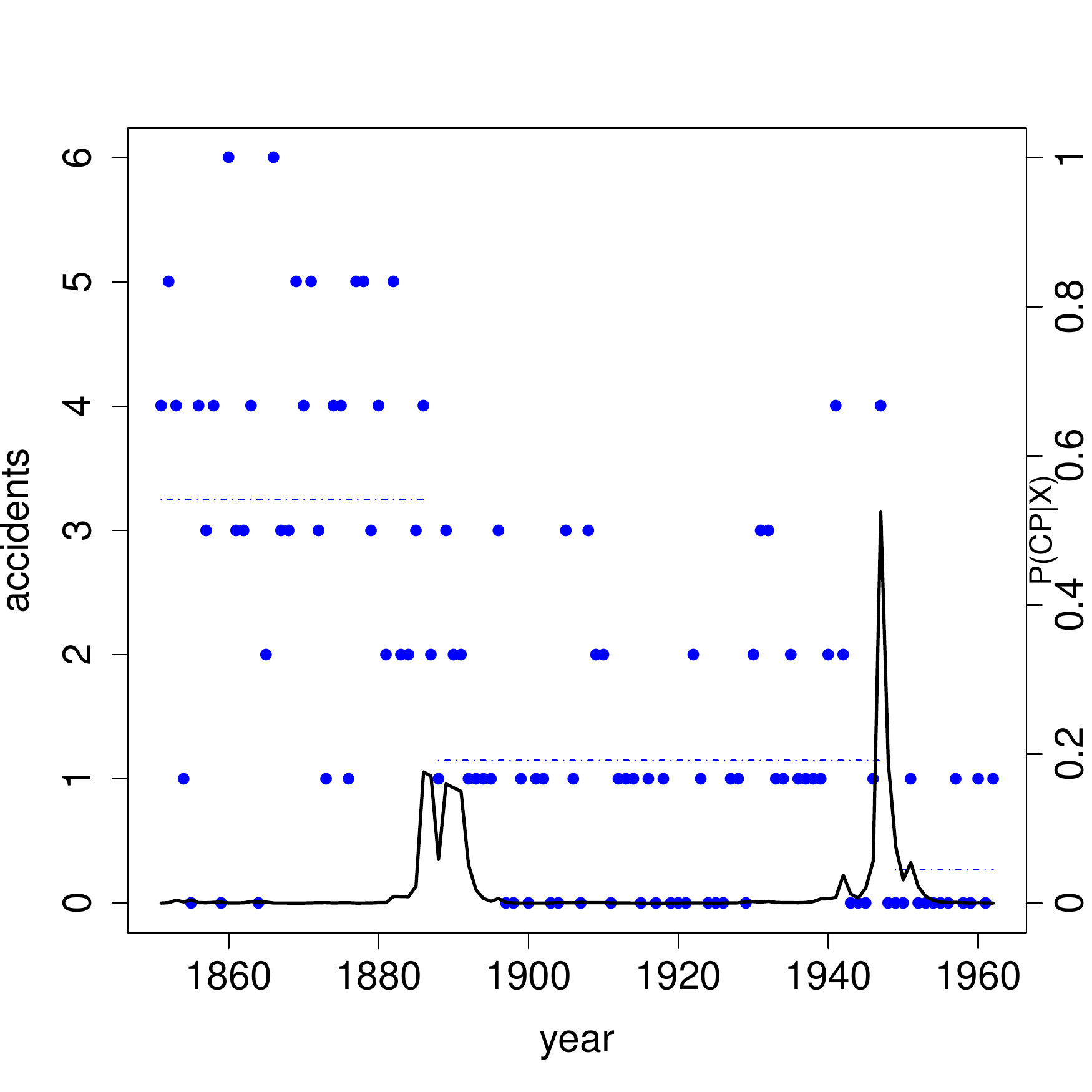}\\
\label{clcp}}    
\caption{\emph{Level-based model.} Plots of estimated posterior probabilities in British coal mining disaster data set \citep{carlin92}. Dots are annual accidents with horizontal lines being the maximum likelihood estimates of the $L=3$ level means. (a) Posterior marginal probability of $i$ being in state $r$, $\mathbb{P}(S_i=r|\mathcal{E}^\mathcal{L})$, with solid line being state $1$, dashed line state $2$ and dashed-dotted line state $3$. (b) Posterior probability of $i$ being any change-point $\mathbb{P}(CP=i|\mathcal{E}^\mathcal{L})=\mathbb{P}(S_i=r,S_{i+1}=s|\mathcal{E}^\mathcal{L})$, where $r\neq s$.}
\subfigure[Posterior state probability]{
\includegraphics[scale=0.35]{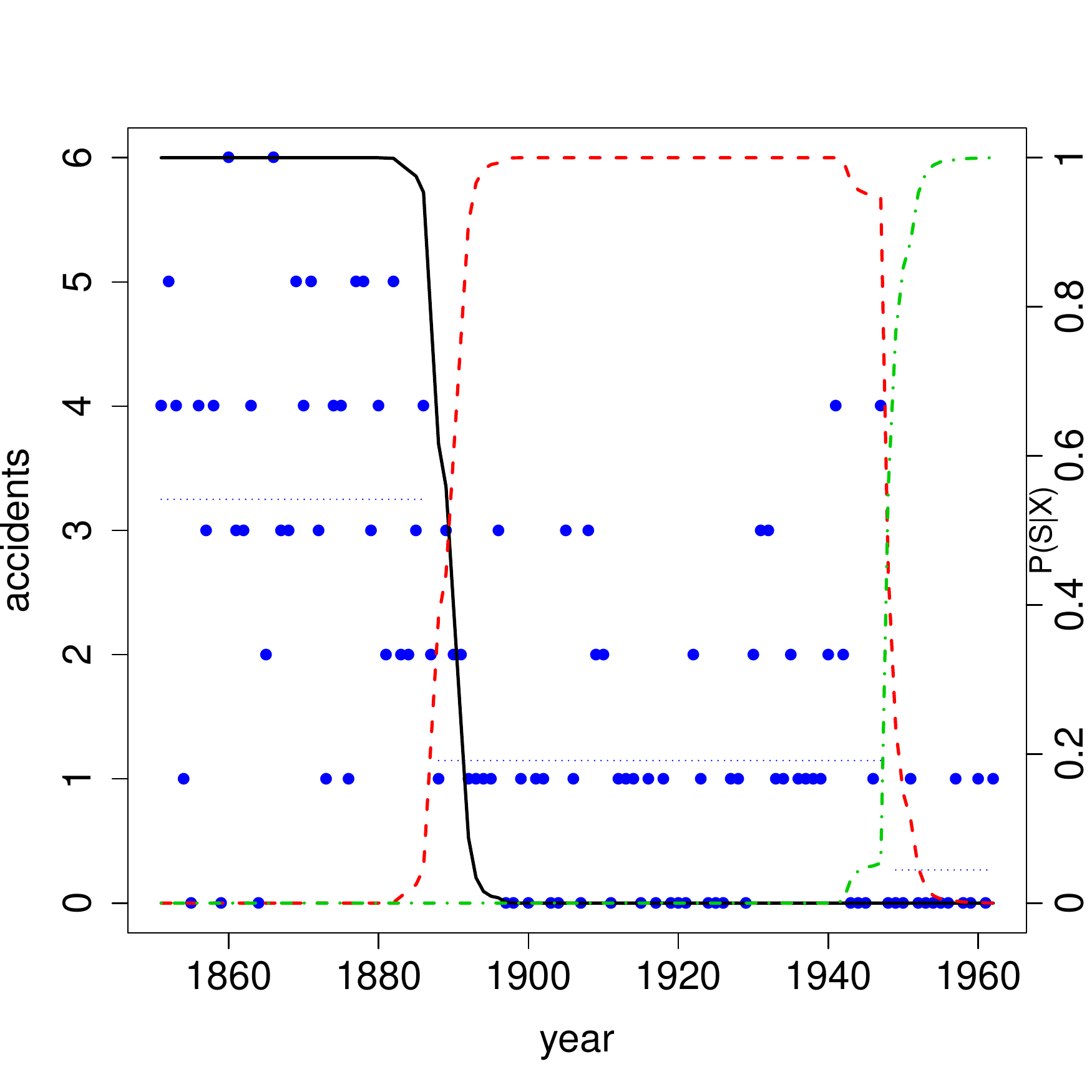}
\label{css}}       
\subfigure[Posterior change-point probability]{
\includegraphics[scale=0.35]{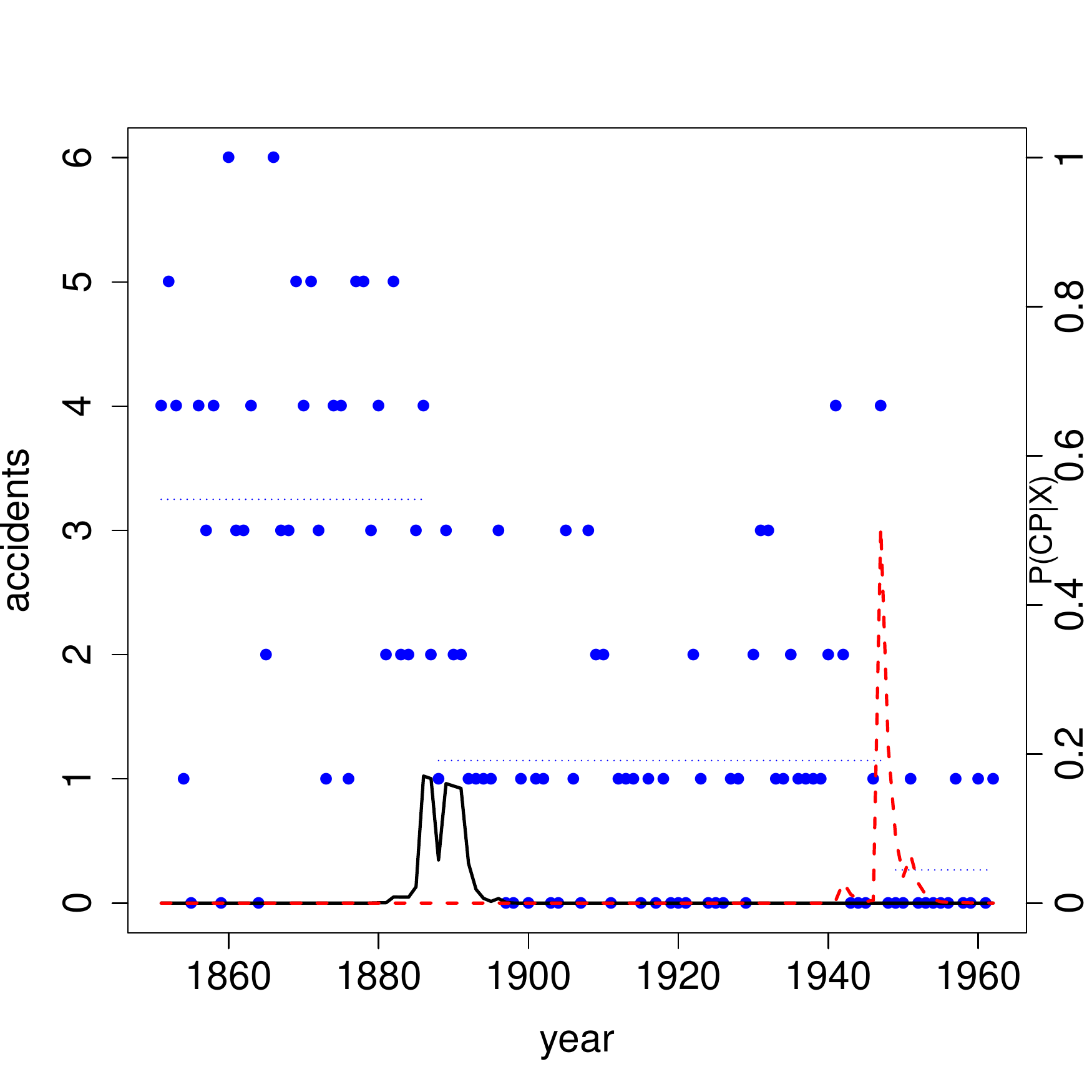}
\label{cscp}}    
\caption{\emph{Segment-based model.} Plots of estimated posterior probabilities in British coal mining disaster data set \citep{carlin92}. Dots are annual accidents with horizontal lines being the maximum likelihood estimates of the $K=3$ segment means. (a) Posterior marginal probability of $i$ being in state $r$, $\mathbb{P}(S_i=r|\mathcal{E}^\mathcal{S})$, with solid line being segment $1$, dashed line segment $2$ and dashed-dotted line segment $3$. (b) Posterior probability of $i$ being $r^{th}$ change-point $\mathbb{P}(CP_r=i|\mathcal{E}^\mathcal{S})=\mathbb{P}(S_i=r,S_{i+1}=r+1|\mathcal{E}^\mathcal{S})$, with solid line being change-point $1$, dashed line change-point $2$.}
\end{figure*}

We illustrate the methods on the classical British coal mining disaster data set \citep{carlin92}, which displays the number of accidents per year in Great Britain between 1851 and 1962, $n=112$. The observed count data in the HMM use a Poisson emission distribution. For a change-point model with 3 segments, a greedy least squares minimization algorithm \citep{hartigan79} identified change-points at $i=36$ and $97$, corresponding to the years $1886$ and $1947$, respectively. 

The level-based approach models the change-points as transitions between states, which correspond to mean parameters of the Poisson distribution, or levels. Using these initial change-points, we obtain maximum likelihood estimates of means $\hat{\lambda}=(3.25,1.15,0.27)$ and the transition probabilities to be $\hat{\eta}=(1/36,1/61)$. Through the forward-backward algorithm described in section \ref{subsec:F/B} the posterior estimates of the state of observation $i$ given the data are $\mathbb{P}(S_i=r|\mathcal{E}^\mathcal{L})$ (Figure \ref{cls}) for each of $L=3$ levels, and for any change occurring at $i$ $\mathbb{P}(CP=i|\mathcal{E}^\mathcal{L})$ (Figure \ref{clcp}).

Figures \ref{css} and \ref{cscp} display the posterior probabilities of the marginal distribution $\mathbb{P}(S_i=r|\mathcal{E}^\mathcal{S})$ for each of $K=3$ segments and change-point probability $\mathbb{P}(CP_r=i|\mathcal{E}^\mathcal{S})=\mathbb{P}(S_i=r,S_{i+1}=r+1|\mathcal{E}^\mathcal{S})$ for $K-1=2$ change-points, respectively. The segment-based approach models the change-points as the beginning and end of $3$ intervals of contiguous observations with homogeneous distributions. Considering the hidden states of both level and segment -based HMMs are the same, the two different approaches produce almost identical respective marginal state and change-point curves.

From a historical perspective it is practical to look for events that may have triggered these detected changes in accident rates. In previous change-point literature, \citet{raftery86} noted that the first change at year $1886$ occurred during a decline in labor productivity at the end of the 1800s and the emergence of the Miners' Federation. Though, the uncertainty in the first change-point location suggests that the reduction in accidents was not due to a sudden event but to a continual decline over this time period. Meanwhile, the second change-point at $1947$ coincides with a more clearly-defined time point that is two years after the end of the second World War.

\subsection{Breast cancer data set}

\begin{figure*}
\subfigure[Posterior marginal state probability]{
\includegraphics[scale=0.35]{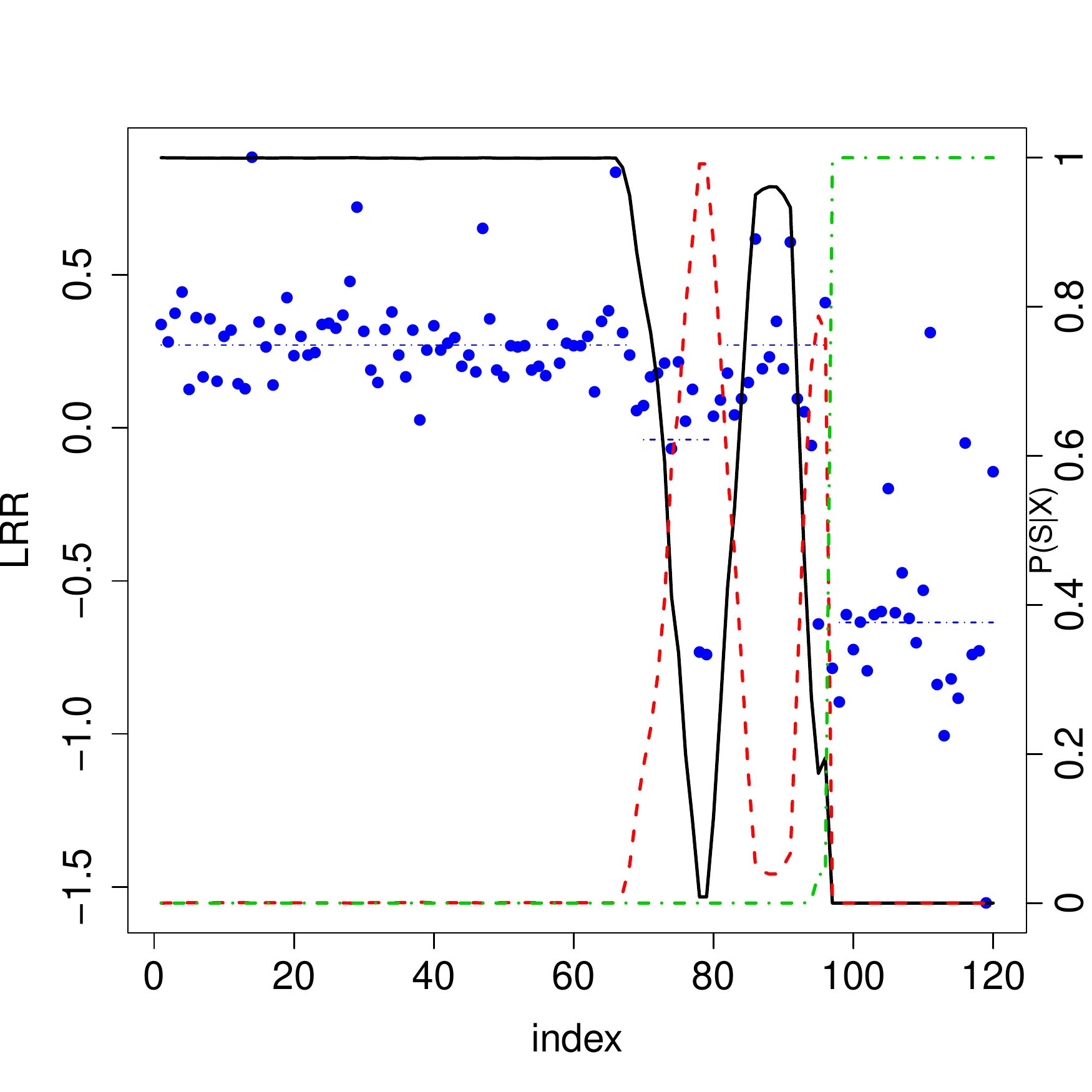}
\label{lls}}       
\subfigure[Posterior change-point probability]{
\includegraphics[scale=0.35]{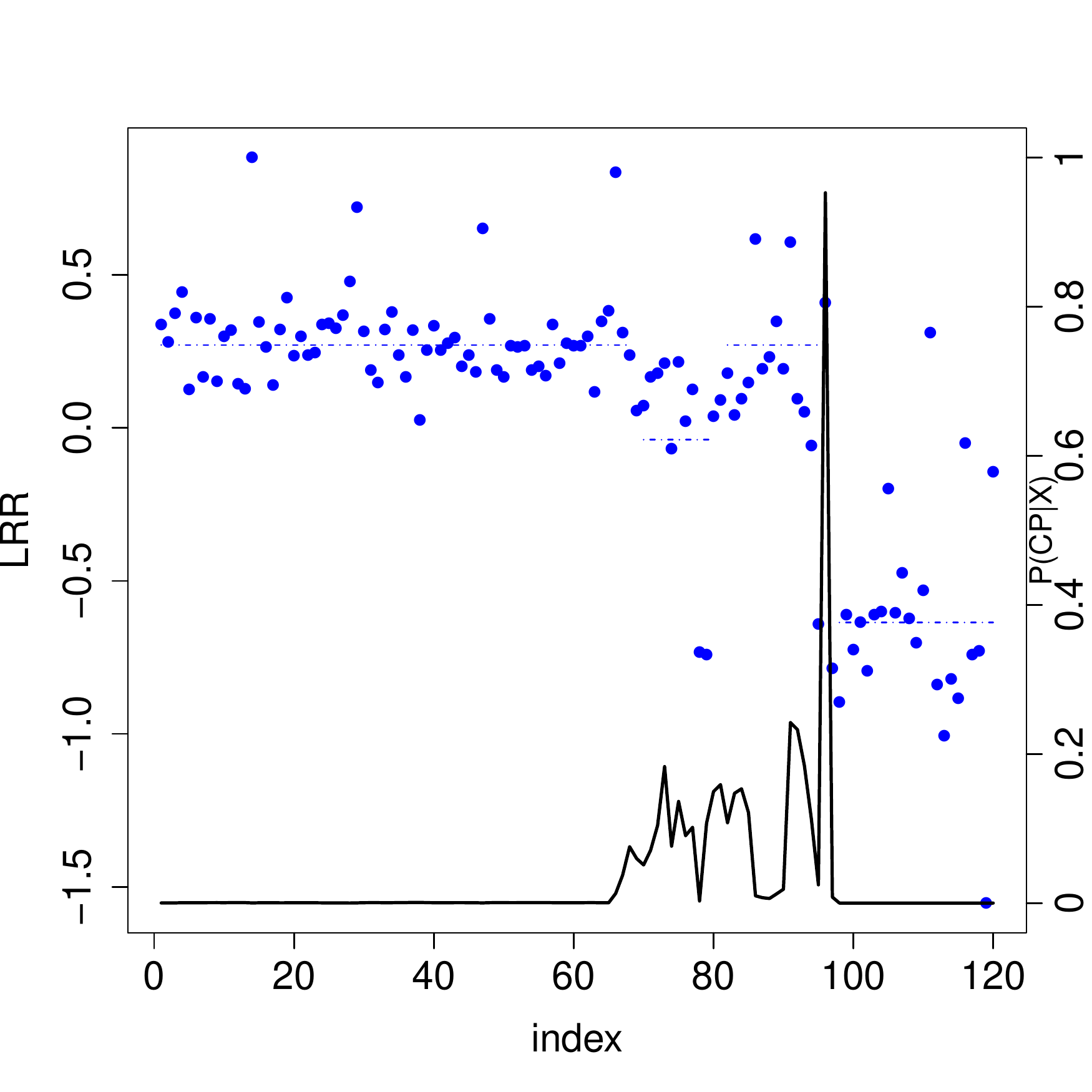}\\
\label{llcp}}    
\caption{\emph{Level-based model.} Plots of estimated posterior probabilities in breast-cancer data set \citep{snijders01}. Dots are LRR with horizontal lines being the maximum likelihood estimates of the $L=3$ level means. (a) Posterior marginal probability of $i$ being in state $r$, $\mathbb{P}(S_i=r|\mathcal{E}^\mathcal{L})$, with solid line being state $1$, dashed line state $2$ and dashed-dotted line state $3$. (b) Posterior probability of $i$ being any change-point $\mathbb{P}(CP=i|\mathcal{E}^\mathcal{L})=\mathbb{P}(S_i=r,S_{i+1}=s|\mathcal{E}^\mathcal{L})$, where $r\neq s$.}
\subfigure[Posterior marginal state probability]{
\includegraphics[scale=0.35]{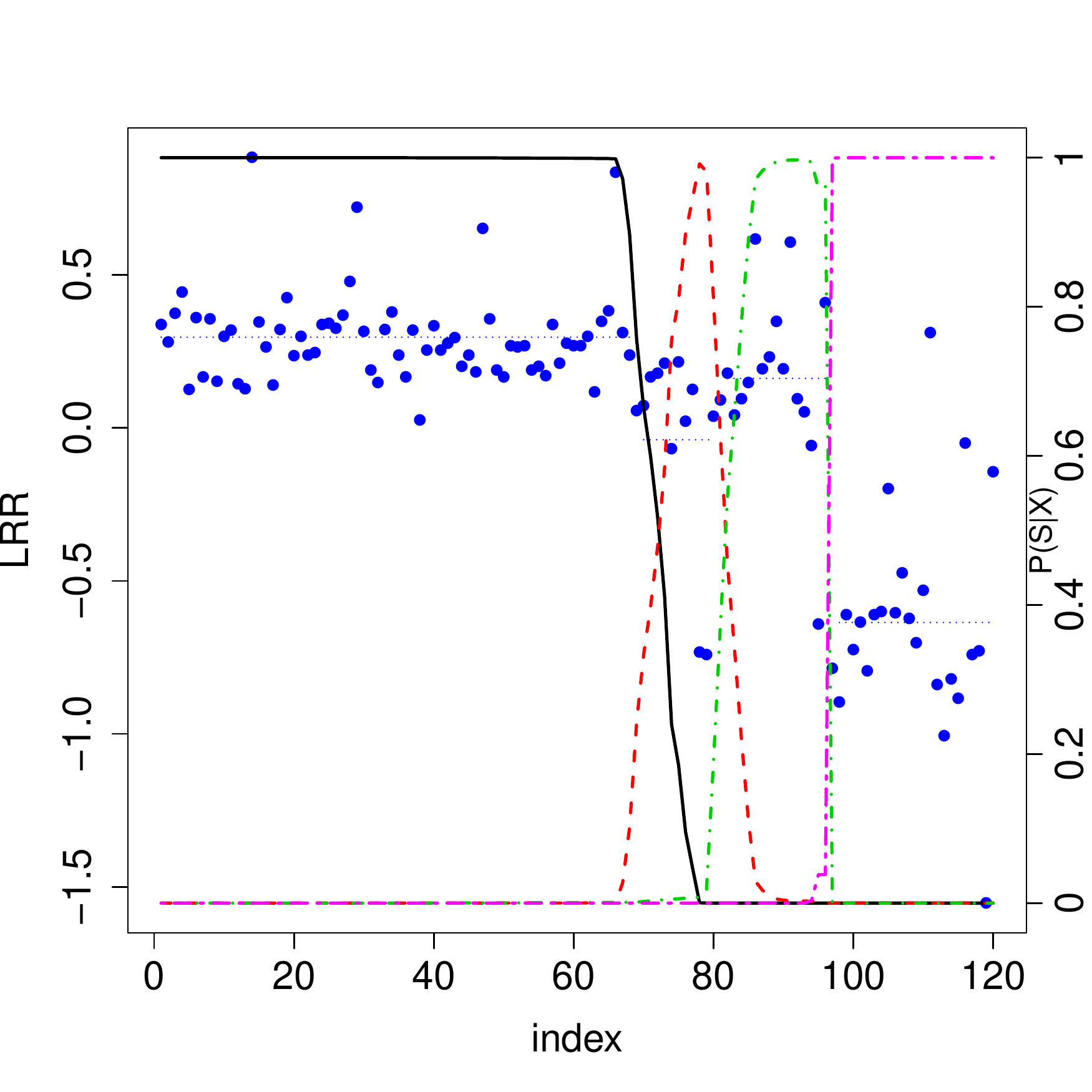}
\label{lss}}       
\subfigure[Posterior change-point probability]{
\includegraphics[scale=0.35]{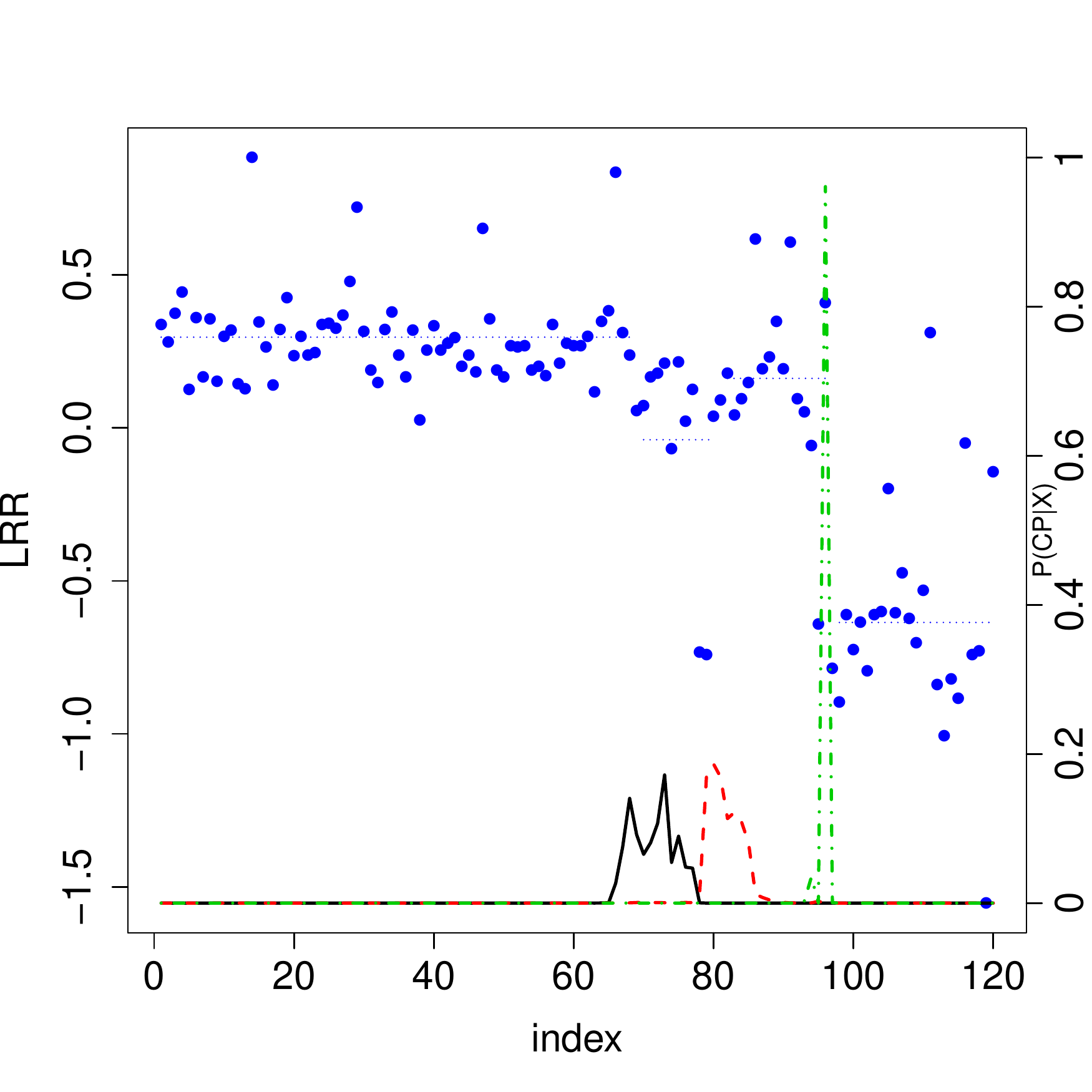}
\label{lscp}}    
\caption{\emph{Segment-based model.} Plots of estimated posterior probabilities in breast-cancer data set \citep{snijders01}. Dots are LRR with horizontal lines being the maximum likelihood estimates of the $K=4$ segment means. (a) Posterior marginal probability of $i$ being in state $r$ $\mathbb{P}(S_i=r|\mathcal{E}^\mathcal{S})$, with solid line being segment $1$, dashed line segment $2$, dashed-dotted line segment $3$ and long dashed line segment $4$. (b) Posterior probability of $i$ being $r^{th}$ change-point $\mathbb{P}(CP_r=i|\mathcal{E}^\mathcal{S})=\mathbb{P}(S_i=r,S_{i+1}=r+1|\mathcal{E}^\mathcal{S})$, with solid line being change-point $1$, dashed line change-point $2$, and dotted-dashed line change-point $3$.}
\end{figure*}

We apply the methods to a widely referenced data set from a breast cancer cell line BT474 \citep{snijders01}. The data consist of log-reference ratios (LRRs) signifying the ratio of genomic copies of test samples compared to normal. The goal is to segment the data into segments with similar copy numbers, with change-points pointing to a copy number aberration that may signify genetic mutations of interest \citep{pinkel98}. We use the same data previously analyzed \citep{rigaill11}, consisting of $n=120$ observations from chromosome~10. The observations are sorted according to their relative position along chromosome~10. The observed LRR data in the HMM uses a normal emission distribution. 

For a change-point model with 4 segments, the least squares algorithm identified the most likely change-points at $i=68,80$ and $96$. Under this specific level-based model, the observations in segments 1 and 3 may share the same distribution. Using these initial change-points, the maximum likelihood estimates are $\hat{\mu}=(0.271,-0.039,-0.636)$ for the $3$ levels and the transition probabilities to be $\hat{\eta}=(2/84,1/16)$. We obtain the posterior estimates of the state of observation $i$ given the data $\mathbb{P}(S_i=r|\mathcal{E}^\mathcal{L})$ (Figure \ref{lls}) for $L=3$ levels and any change occurring at $i$ $\mathbb{P}(CP=i|\mathcal{E}^\mathcal{L})$ (Figure \ref{llcp}).

Figures \ref{lss} and \ref{lscp} display the posterior probabilities of the marginal distribution $\mathbb{P}(S_i=r|\mathcal{E}^\mathcal{S})$ for $K=4$ segments and change-point probabilities $\mathbb{P}(CP_r=i|\mathcal{E}^\mathcal{S})=\mathbb{P}(S_i=r,S_{i+1}=r+1|\mathcal{E}^\mathcal{S})$ for $K-1=3$ change-points, respectively. The maximum likelihood estimates of means are $\hat{\mu}=(0.289,-0.039,0.224,-0.636)$ for the 4 segments. With different segments $1$ and $3$ defined as homogeneous, the two approaches produce slightly different probability distributions of change-point locations, though the locations of peaks remain similar.

\begin{figure*}
\centering{
\includegraphics[scale=0.5]{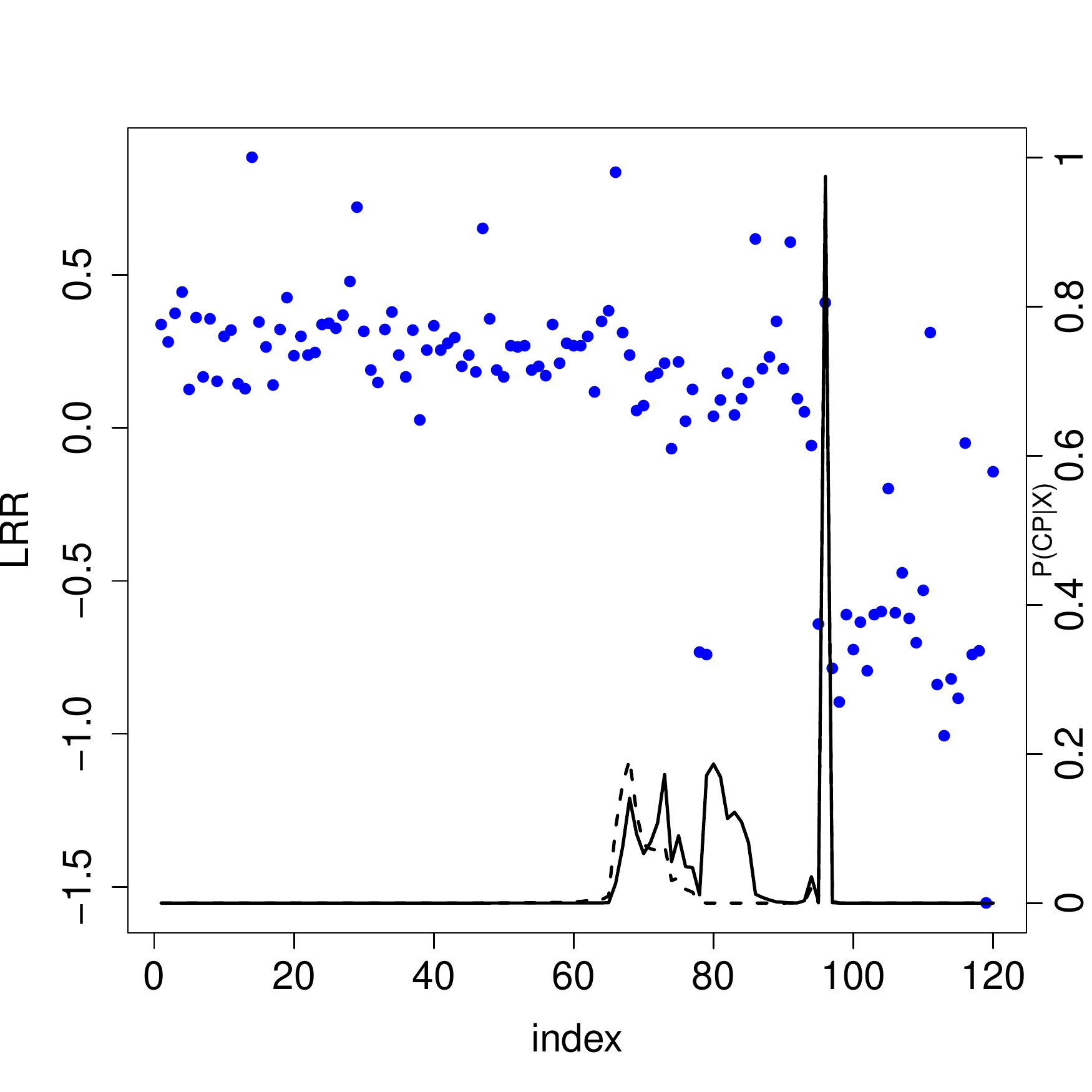}
\caption{Plots of estimated posterior change-point probabilities of $K=3$ and $K=4$ segment change-point models in breast-cancer data set \citep{snijders01}. Dots are LRR, lines are posterior probabilities of a change-point, where $\mathbb{P}(CP_r=i|\mathcal{E}^\mathcal{S})=\mathbb{P}(S_i=r,S_{i+1}=r+1|\mathcal{E}^\mathcal{S})$. Change-point probability curve is dotted for $K=3$ segment model and solid for $K=4$ segment model.}\label{l3v4}    }
\end{figure*}

Figure \ref{l3v4} compares the change-point probability curves for both $K=3$ and $K=4$ segment-based change-point models. The $K=3$ segment-based model does not include the second change-point at $i=80$. The shape of the posterior probability curve of the first change-point slightly changes between the two models, due to the uncertainty in the precise location of this point. On the other hand, the peak of the last change-point is close to one, and due to this high precision, the corresponding change-point curve from both the $K=3$ and $K=4$ change-point models virtually overlap.

\section{Current implementations of change-point models for genomics data}\label{genomics}

A common point of interest in bioinformatics is to find genetic mutations pointing to phenotypes susceptible in cancer and other diseases. The detection of change points in Copy Number Variation (CNV) is a critical step in the characterization of DNA, including tumoral DNA in cancer. A CNV may locate a genetic mutation such as a duplication or deletion in a cancerous cell that is a target for treatment. 

Many level-based approaches use the hidden state space in a classical discrete HMM to characterize mutations \citep{fridlyand04} to map the number of states and the most likely state at each position. Various extensions to this HMM approach include various procedures such as merging change-points and specifying prior transition matrices \citep{willenbrock05,marioni06} to improve the results. Other extensions include reversible-jump Markov chain Monte Carlo (MCMC) to fit the HMM \citep{rueda07}, a continuous-index HMM \citep{stjernqvist07} that takes into account the discrete nature of observed genomics data, and methods that take into account genomic distance and overlap between clones \citep{andersson08}. This HMM approach has also been extended for simultaneous change-points across multiple samples \citep{shah09} and for simultaneously identifying multiple outcomes \citep{liu10}. Other HMM-based implementations \citep{colella07,wang07} deal with higher-resolution data for current SNP array technologies. 

A wide amount of non-HMM approaches are also available for bioinformatics data. The aim of these approaches is typically in finding contiguous segments consisting of observations with the same distribution. One such implementation is a non-parametric extension of binary segmentation for multiple change-points through permutation \citep{olshen04} along with a faster extension \citep{venkatraman07} which uses stopping rules. Various smoothing techniques \citep{hupe04,eilers05,hsu05} have also been applied to change-point modelling. Summaries and comparisons of the various approaches for finding change-points in genomics data are available \citep{willenbrock05,lai05}.


\section{Conclusion}\label{conclusion}

This chapter describes a simple algorithm using hidden Markov models to estimate posterior distributions of interest in change-point analysis, using two different modelling approaches. It also addresses computational issues through several simple constraints on the HMM to allow for estimates in a feasible amount of time. 

HMMs are a special case of Bayesian Networks (BNs), which are probabilistic graphical models that represent random variables and their dependencies via a Directed Acyclic Graph (DAG), for example in Figure~\ref{fig:hmm}. The conditional dependencies coded by the edges of the DAG determine a factorization of the joint probability distribution. Given any type of evidence, summing variables out in the joint distribution obtains the conditional distributions of the variables. For BNs, this is done in a precise fashion by the exact Belief Propagation (BP) algorithm, which efficiently applies the distributive law to provide a recursive decomposition of the initial sums of products. The intermediate sums of products (the so-called messages or beliefs) propagate through a secondary tree-shaped structure called a \emph{junction tree} or \emph{graph tree}. This general framework allows for recursive algorithms to obtain various quantities of interest including posterior distribution, likelihood or entropy-related terms. Moreover, the recursive algorithms for BNs permit model selection through the use of typical criteria such as the Bayesian Information Criterion \citep[BIC, see][]{schwarz1978estimating} or the Integrated Completed Likelihood \citep[ICL, see][]{biernacki2000assessing}. 

The exact BP general algorithm makes it possible to easily derive forward and backward recursions for HMMs for any kind of evidence. Moreover the BN unifying framework permits simple extensions of the aforementioned change-point HMMs to more complex and realistic models accounting for intricate dependencies between the variables or data. As an example, a straightforward generalization of the aforementioned segment- and level-based models is to account simultaneously for both hidden segments and levels. More specifically, let $S_i$ and $L_i$ be the segment and level, respectively, of observation $X_i$. In this specification, the current level $L_i$ not only depends on $L_{i-1}$, but also on both the hidden states $S_i$ and $S_{i-1}$. 

Another useful extension is a model with joint segmentation of multiple samples. In a two-sample situation, there are two sets of distinct observations $X_1,\ldots,X_n$ and $Y_1,\ldots,Y_n$, one for each sample. Let $U_n,\ldots,U_n$ be the hidden variables corresponding to the segments of the first sample and $V_1,\ldots,V_n$ the variables corresponding to the segments of the second sample. The model assumes that $U_1,\ldots,U_n$ and $V_1,\ldots,V_n$ both depend on a common segmentation $S_1,\ldots,S_n$. An advantage of this model is that the joint segmentation of both samples is not independent and can identify common features of both samples.

\bibliographystyle{plainnat}
\bibliography{cp}

\appendix

\section{Sample R code}

The following R code presents a toy example for data simulated from the Poisson distribution, with $n=100$ observations and $J=3$ different segments, with change-points after the observations $25$ and $75$. The previously described algorithms calculate the probability of an observation being in a hidden state (marginal) and being a change-point (cp).

\newpage
The code for the \emph{level-based} models:
\begin{lstlisting}
n=100; # not usable for large n due to underflow issue 
(logscale version needed for larger n)
L=3; # number of levels
lambda=c(4.5,3.0,7.0); # parameters for the three level
eta=0.03; # transition parameter
# transition matrix
pi=matrix(eta/(L-1),ncol=L,nrow=L);
diag(pi)=1.0-eta;

# generate the data according to the model
set.seed(42)
s=rep(NA,n); s[1]=1;
for (i in 2:n) s[i]=sample(1:L,size=1,prob=pi[s[i-1],]);
x=rpois(n,lambda[s]);
refcp=which(diff(s)!=0); # reference change-point location

# forward recursion
F=matrix(rep(NA,n*L),ncol=n);
F[,1]=0.0; F[1,1]=1.0*dpois(x[1],lambda[1]);
for (i in 2:n) for (l in 1:L) 
 F[l,i]=sum(F[,i-1]*pi[,l]*dpois(x[i],lambda[l]));

# backward recursion
B=matrix(rep(NA,n*L),ncol=n); B[,n]=1.0; 
for (i in seq(n,2,by=-1)) for (l in 1:l)  
 B[l,i-1]=sum(pi[l,]*dpois(x[i],lambda)*B[,i]);

# probability of the evidence
pevidence=F[1,1]*B[1,1];

# consistency check TRUE if OK
# NOT RUN: prod(apply(F*B,2,sum)==pevidence)==1

# marginal distribution
marginal=F*B/pevidence;

# posterior distribution of change-points
cp=rep(0,n);
for (i in 2:n) for (l in 1:L) 
  cp[i-1]=cp[i-1]+sum(F[l,i-1]*pi[l,-l]*dpois(x[i],lambda[-l])*B[-l,i])/pevidence;

par(mfrow=c(1,2));
plot(x,main="posterior marginal distribution of segment index",pch=16,col="blue"); 
abline(v=refcp,col="blue",lty=4); points(lambda[s],t="l",col="blue",lty=4);
for (j in 1:L) points(max(x)*marginal[j,],col=j,lty=j,t="l",lwd=2);
plot(x,main="posterior change-point distribution",pch=16,col="blue"); 
abline(v=refcp,col="blue",lty=4); 
points(lambda[s],t="l",col="blue",lty=4);
points(max(x)*cp,t="l",lwd=2);
\end{lstlisting}

\newpage
The code for the \emph{segment-based} models:
\begin{lstlisting}
n=100; # not usable for large n due to underflow issue 
#(logscale version needed for larger n)
K=3; # number of segments
lambda=c(4.5,3.0,7.0); # parameters for the three segments
refcp=c(25,75); # ref change point locations
s=c(rep(1,25),rep(2,50),rep(3,25)); # true segment index

# generate the data according to the model
set.seed(42)
x=rpois(n,lambda[s]);

# forward recursion
F=matrix(rep(NA,n*K),ncol=n);
F[,1]=0.0; F[1,1]=1.0*dpois(x[1],lambda[1]);
for (i in 2:n) {
  F[2:K,i]=(0.5*F[1:(K-1),i-1]+0.5*F[2:K,i-1])*dpois(x[i],lambda[2:K]);
  F[1,i]=0.5*F[1,i-1]*dpois(x[i],lambda[1]);
};

# backward recursion
B=matrix(rep(NA,n*K),ncol=n); B[,n]=0.0; B[K,n]=1.0;
for (i in seq(n,2,by=-1)) {
  B[1:(K-1),i-1]=0.5*B[1:(K-1),i]*
    dpois(x[i],lambda[1:(K-1)])+0.5*B[2:K,i]*dpois(x[i],lambda[2:K]);
  B[K,i-1]=0.5*B[K,i]*dpois(x[i],lambda[K]);
};

# probability of the evidence
pevidence=F[1,1]*B[1,1];

# consistency check TRUE if OK
# NOT RUN: prod(apply(F*B,2,sum)==pevidence)==1

# marginal distribution
marginal=F*B/pevidence;

# posterior distribution of change-points
cp=matrix(0,nrow=K-1,ncol=n);
for (k in 1:(K-1)) cp[k,1:(n-1)]=F[k,1:(n-1)]/pevidence*0.5*B[k+1,2:n]*
  dpois(x[2:n],lambda=lambda[k+1]);

par(mfrow=c(1,2));
plot(x,main="posterior marginal distribution of segment index",pch=16,col="blue"); 
abline(v=refcp,col="blue",lty=4); points(lambda[s],t="l",col="blue",lty=4);
for (k in 1:K) points(max(x)*marginal[k,],col=k,lty=k,t="l",lwd=2);
plot(x,main="posterior change-point distribution",pch=16,col="blue"); 
abline(v=refcp,col="blue",lty=4); 
points(lambda[s],t="l",col="blue",lty=4);
for (k in 1:(K-1)) points(max(x)*cp[k,],col=k,lty=k,t="l",lwd=2);
\end{lstlisting}

\newpage
\section{Sample datasets}
 
The number of annual accidents from $1851-1962$ in the British coal mining disaster data set \citep{carlin92}:
\begin{lstlisting}
accidents <- c(4, 5, 4, 1, 0, 4, 3, 4, 0, 6, 3, 3, 4, 0, 2, 6, 3, 3, 5, 
4, 5, 3, 1, 4, 4, 1, 5, 5, 3, 4, 2, 5, 2, 2, 3, 4, 2, 1, 3, 2, 2, 1, 1, 
1, 1, 3, 0, 0, 1, 0, 1, 1, 0, 0, 3, 1, 0, 3, 2, 2, 0, 1, 1, 1, 0, 1, 0, 
1, 0, 0, 0, 2, 1, 0, 0, 0, 1, 1, 0, 2, 3, 3, 1, 1, 2, 1, 1, 1, 1, 2, 4, 
2, 0, 0, 0, 1, 4, 0, 0, 0, 1, 0, 0, 0, 0, 0, 1, 0, 0, 1, 0, 1)\end{lstlisting}

The log-reference ratio (LRR) of chromosome~10 of cell line BT474 in breast cancer tumor \citep{snijders01}:
\begin{lstlisting}
LRR <- c(0.3362, 0.2793, 0.3742, 0.4424, 0.1238, 0.3590, 0.1655, 0.3552, 
0.1504, 0.2983, 0.3173, 0.1428, 0.1276, 0.8824, 0.3438, 0.2642, 0.1390, 
0.3211, 0.4235, 0.2338, 0.2983, 0.2376, 0.2452, 0.3362, 0.3400, 0.3248, 
0.3666, 0.4766, 0.7193, 0.3135, 0.1883, 0.1466, 0.3211, 0.3779, 0.2376, 
0.1655, 0.3173, 0.0252, 0.2528, 0.3324, 0.2528, 0.2755, 0.2945, 0.1997, 
0.2376, 0.1807, 0.6510, 0.3552, 0.1883, 0.1655, 0.2680, 0.2642, 0.2680, 
0.1883, 0.1997, 0.1693, 0.3362, 0.2111, 0.2755, 0.2680, 0.2680, 0.2983, 
0.1162, 0.3476, 0.3817, 0.8331, 0.3097, 0.2376, 0.0556, 0.0707, 0.1655, 
0.1769, 0.2111, -0.0696, 0.2149, 0.0214, 0.1238, -0.7333, -0.7409, 
0.0366, 0.0897, 0.1769, 0.0404, 0.0935, 0.1466, 0.6169, 0.1921, 0.2300, 
0.3476, 0.1921, 0.6055, 0.0935, 0.0518, -0.0582, -0.6423, 0.4083, 
-0.7864, -0.8964, -0.6120, -0.7258, -0.6347, -0.7940, -0.6120, -0.6006,
-0.1986, -0.6044, -0.4754, -0.6234, -0.7030, -0.5323, 0.3097, -0.8395, 
-1.0064, -0.8206, -0.8851, -0.0506, -0.7409, -0.7296, -1.5526, -0.1455)
\end{lstlisting}

\end{document}